\documentclass[conference,a4paper]{IEEEtran}

\usepackage{amsfonts}
\usepackage{amssymb}
\usepackage{amsmath}
\usepackage{bbm}
\usepackage{bm}
\usepackage{dsfont}
\usepackage{graphicx}

\newcommand{\bra}[1]{#1^\dagger}
\newcommand{\ket}[1]{#1}
\newcommand{\braket}[2]{#1^\dagger #2}
\newcommand{\myu}{\bm{u}}
\newcommand{\myc}{\bm{c}}
\newcommand{\myf}{\bm{f}}
\newcommand{\myF}{\bm{F}}
\newcommand{\mypsi}{\bm{\psi}}
\newcommand{\myPsi}{\bm{\Psi}}

\newcommand{\myMat}{\bm{\Phi}}

\newtheorem{theorem}{Theorem}

\newtheorem{corollary}{Corollary}

\begin{document}

\sloppy

\title{An ``Umbrella'' Bound of the Lov\'asz-Gallager Type}

\author{
  \IEEEauthorblockN{Marco Dalai}
  \IEEEauthorblockA{Department of Information Engineering\\
    University of Brescia - Italy\\
    Email: marco.dalai@ing.unibs.it} 
}



\maketitle

\begin{abstract}
We propose a novel approach for bounding the probability of error of discrete memoryless channels with a zero-error capacity based on a combination of Lov\'asz' and Gallager's ideas. The obtained bounds are expressed in terms of a function $\vartheta(\rho)$, introduced here, that varies from the cut-off rate of the channel to the Lov\'azs theta function as $\rho$ varies from $1$ to $\infty$ and which is intimately related to Gallager's expurgated coefficient. 
The obtained bound to the reliability function, though loose in its present form, is finite for all rates larger than the Lov\'asz theta function.
\end{abstract}

\section{Introduction}

One of the most intriguing topic in coding theory is the problem of bounding the probability of error of optimal codes at low rates. While at high rates the asymptotic behaviour of the probability of error for optimal codes is now very well understood, very little is known in the low rate region.
Shannon \cite{shannon-1948} introduced the notion of channel capacity $C$, which  represents the largest rate at which information can be sent through the channel with probability of error that vanishes with increasing block-length. He then also introduced \cite{shannon-1956} the notion of zero-error capacity $C_0$ as the largest rate at which information can be sent with probability of error precisely equal to zero. For rates in the range $C_0<R<C$, the probability of error is known to decrease exponentially in the block-length $n$ as
\begin{equation}
P_e\approx e^{-n E(R)},
\end{equation}
where $E(R)$ is the so called reliability function of the channel.
Both determining $E(R)$ for small $R$ and even determining $C_0$ is an unsolved problem and only upper and lower bounds for these quantities are known. Lov\'asz gave an important improvement in upper bounding $C_0$ by means of his $\vartheta$ function, thus enlarging the range of values over which $E(R)$ is known to be finite. However, Lov\'asz's result was never exploited to find actual bounds to $E(R)$ for rates $R$ immediately above $\vartheta$.

In this paper we propose a first attempt to bound the probability of error at rates $R>\vartheta$ by combining Lov\'asz's method  with the idea used by Gallager in the development of his expurgated bound.
Even if the proposed approach does not lead yet to good bounds to the reliability function, we believe it sheds some light on this relatively unexplored topic.


\section{Basic notions}
\subsection{Reliability of DMCs}
Let $W(x|y)$, $x\in \mathcal{X}$, $y\in\mathcal{Y}$, be the transition probabilities of a discrete memoryless channel $W :\mathcal{X}\to\mathcal{Y}$, where $\mathcal{X}=\{1,2,\ldots,K\}$ and $\mathcal{Y}=\{1,2,\ldots,J\}$ are finite sets. For a sequence $\bm{x}=(x_1,x_2,\ldots,x_n)\in\mathcal{X}^n$ and a sequence $\bm{y}=(y_1,y_2,\ldots,y_n)\in\mathcal{Y}^n$, the probability of observing $\bm{y}$ at the output of the channel given $\bm{x}$ at the input is 
\begin{equation}
W^{(n)}(\bm{y}|\bm{x})=\prod_{i=1}^n W(y_i|x_i).
\end{equation}
A block code with $M$ messages and block-length $n$ is a mapping from a set $\{1,2,\ldots,M\}$ of $M$ messages onto a set $\{\bm{x}_1, \bm{x}_2, \ldots, \bm{x}_M\}$ of $M$ sequences in $\mathcal{X}^n$. The rate $R$ of the code is defined as $R=\log M/n$.
A decoder is a mapping from $\mathcal{Y}^n$ into the set of possible messages $\{1,2,\ldots,M\}$. If message $m$ is to be sent, the encoder transmits  the codeword $\bm{x}_m$ through the channel. An output sequence $\bm{y}$ is received by the decoder, which maps it to a message $\hat{m}$. An error occurs if $\hat{m}\neq m$.

Let $Y_m\subseteq \mathcal{Y}^n$ be the set of output sequences that are mapped into message $m$. When message $m$ is sent, the probability of error is
\begin{equation}
P_{e|m}=\sum_{\bm{y}\notin Y_m} W^{(n)}(\bm{y}|\bm{x}_m).
\end{equation}
The maximum error probability of the code is defined as the largest $P_{e|m}$, that is,
\begin{equation}
P_{e,\max}=\max_{m}P_{e|m}.
\end{equation}

Let $P_{e,\max}^{(n)}(R)$ be the smallest maximum error probability among all codes of length $n$ and rate at least $R$.
Shannon's theorem \cite{shannon-1948} states that sequences of codes exist such that $P_{e,\max}^{(n)}(R)\to 0$ as $n\to\infty$ for all rates smaller than a constant $C$, called \emph{channel capacity}.
For $R<C$, Shannon's theorem only asserts that $P_{e,\max}^{(n)}(R)\to 0$ as $n\to\infty$.  
For a range of rates  $C_0< R < C$, the optimal probability of error $P_{e,\max}^{(n)}(R)$ is known to have an exponential decrease in $n$, and it is thus useful to define the \emph{reliability function} of the channel as
\begin{equation}
E(R)=\limsup_{n\to\infty} -\frac{1}{n}\log P_{e,\max}^{(n)}(R).
\label{eq:E(R)_def_class}
\end{equation} 
The value $C_0$ is the so called \emph{zero-error capacity}, also introduced by Shannon \cite{shannon-1956}, which is defined as the highest rate at which communication is possible with probability of error precisely equal to zero. More formally,
\begin{equation}
C_0=\sup\{R \, :\,   P_{e,\max}^{(n)}(R)=0 \mbox{ for some } n\}.
\end{equation}
For $R<C_0$, we may define the reliability function $E(R)$ as being infinite. Note that $C_0>0$ if and only if there are at least two input symbols $x$ and $x'$ which are not confusable at the output, meaning that $W(y|x)W(y|x')$ is zero for all values of $y$.
Determining the reliability function $E(R)$ (at low positive rates) and the zero-error capacity $C_0$ of a general channel is still an unsolved problem.

One of the most famous results in this direction is Lov\'asz's upper bound to $C_0$. 
Lov{\'a}sz proves that $C_0$ is upper bounded by a quantity $\vartheta$ defined as
\begin{align*}
\vartheta& = \min_{\{\myu_x\}}\min_{\myc} \max_x \log \frac{1}{\,|\braket{\myu_x }{\myc}|^2}
\end{align*}
where $\{\myu_x\}_{x\in\mathcal{X}}$ runs over all sets of unit norm vectors in any Hilbert space such that $\myu_x$ and $\myu_{x'}$ are orthogonal if symbols $x$ and $x'$ are not confusable and $\myc$ runs over all unit norm vectors. Here,  $\cdot^\dagger$ denotes conjugate transpose and   $\braket{\myu_x }{\myc}$ is the scalar product between $\ket{\myu_x}$ and $\ket{\myc}$.

\subsection{Bhattacharyya distances and scalar products}
\label{sec:Umbrella-Bhattacharyya}
Here, we briefly recall some important connections between the reliability function $E(R)$ and the Bhattacharyya distance between codewords. 
This connection is of great importance since the Bhattacharyya distance between distributions is related to a scalar product between unit norm vectors in a Hilbert space. It is this property that creates an underlying common substrate for Lov\'asz's approach and for bounding the reliability function.
 
For a generic input symbol $x$, consider the unit norm $|\mathcal{Y}|$-dimensional column vector $\mypsi_x$ with components $\mypsi_x(y)=\sqrt{W(y|x)}$. We call this the \emph{state vector} of input symbol $x$, in analogy with the input signals of pure-state classical-quantum channels (see comment at the end of Section \ref{sec:connections}). In the same way, for an input sequence $\bm{x}=(x_1,x_2,\ldots,x_n)$, consider the unit norm $|\mathcal{Y}|^n$-dimensional column vector $\myPsi_{\bm{x}}$ whose components are the values
$\sqrt{W^{(n)}(\bm{y}|\bm{x})}$, that is, $\myPsi_{\bm{x}}$ is simply the element-wise square root of the conditional output distribution given the input sequence $\bm{x}$. Then, since the channel is memoryless, we can write
\begin{equation}
\myPsi_{\bm{x}}=\mypsi_{x_1}\otimes\mypsi_{x_2}\otimes\cdots\mypsi_{x_n}
\label{eq:defPsi}
\end{equation}
where $\otimes$ is the Kronecker product. Let for ease of notation $\myPsi_m$ be the state vector of the codeword $\bm{x}_m$; then we can represent our code $\{\bm{x}_1, \bm{x}_2, \ldots, \bm{x}_M\}$ by means of their associated state vectors $\{\myPsi_1, \myPsi_2, \ldots, \myPsi_M\}$. 
Since all square roots are taken positive, note that our channel has a positive zero-error capacity if and only if there are at least two state vectors $\mypsi_x$, $\mypsi_{x'}$ such that $\braket{\mypsi_x}{ \mypsi_{x'}}=0$. This implies that codes can be 
built such that $\braket{\myPsi_m}{\myPsi_{m'}}=0$ for some $m$, $m'$, that is, the two codewords $m$ and $m'$ cannot be confused at the output. However, the scalar product 
$\braket{\myPsi_m}{\myPsi_{m'}}$ plays a more general role since it is related to the so called Bhattacharyya distance between the two codewords $m$ and $m'$. In particular, in a binary hypothesis testing between codeword $m$ and $m'$, an extension of the Chernoff Bound allows to assert that the minimum error probability asymptotically satisfies \cite{shannon-gallager-berlekamp-1967-2}
\begin{equation}
P_e \doteq \min_{0\leq s\leq 1}\sum_{\bm{y}}W^{(n)}(\bm{y}|\bm{x}_m)^{1-s}W^{(n)}(\bm{y}|\bm{x}_{m'})^s
\end{equation}
where $\doteq$ means equivalence to the first order in the exponent.
For $s=1/2$, the sum above obviously equals $\braket{\myPsi_m}{\myPsi_{m'}}$.
It is easily shown that the minimum above is always between $(\braket{\myPsi_m}{\myPsi_{m'}})^2$ and $\braket{\myPsi_m}{\myPsi_{m'}}$, and it equals the latter for a class of channels, called pairwise reversible channels, that have some symmetry with respect to the input symbols\footnote{Somehow tautologically, pairwise reversible channels are those for which the minimum is achieved for $s=1/2$.} \cite{shannon-gallager-berlekamp-1967-2}. Obviously, for a given code, the probability of error $P_{e,\max}$ is lower bounded by the probability of error in each binary hypothesis test between two codewords. Hence, we find that $P_{e,\max}$ asymptotically satisfies
\begin{equation}
-\frac{1}{n}\log P_{e,\max} \leq -\frac{2}{n} \log \max_{m\neq m'} (\braket{\myPsi_m}{\myPsi_{m'}}) + o(n),
\label{eq"pemaxprodscal1}
\end{equation}
where the coefficient 2 can be removed if the channel is pairwise reversible.
It is thus obvious that it is possible to upper bound $E(R)$ by lower bounding the quantity
\begin{equation}
\gamma=\max_{m\neq m'} \braket{\myPsi_m}{\myPsi_{m'}}
\end{equation}
Lov\'asz's work aims at finding a value $\vartheta$ as small as possible that allows to conclude that, for a set of $M=e^{nR}>e^{n\vartheta}$ codewords, $\gamma$ cannot be zero, and thus at least two codewords are confusable. Here, instead, we want something more, that is, finding a lower bound on $\gamma$ for each code with rate $R>\vartheta$ so as to deduce an upper bound to $E(R)$ for all $R>\vartheta$.

\section{An ``umbrella'' bound}
\label{sec:Umbrella-umbrella}
Consider the scalar products between the channel state vectors $\braket{\mypsi_x}{\mypsi_{x'}}\geq 0$. For a fixed $\rho\geq 1$, consider then a set of ``tilted'' state vectors, that is, unit norm vectors $\{\tilde{\mypsi}_x\}$ in any Hilbert space such that $|\braket{\tilde{\mypsi}_x}{\tilde{\mypsi}_{x'}}|\leq (\braket{{\mypsi}_x}{{\mypsi}_{x'}})^{1/\rho}$. We call such a set of vectors $\{\tilde{\mypsi}_x\}$ an \emph{orthonormal representation of degree $\rho$} of our channel, and call $\Gamma(\rho)$ the set of all possible such representations\\
\begin{equation}
\Gamma(\rho) = \left\{ \{\tilde{\mypsi}_x\} \, :\,  |\braket{\tilde{\mypsi}_x}{\tilde{\mypsi}_{x'}}|\leq (\braket{{\mypsi}_x}{{\mypsi}_{x'}})^{1/\rho}\right\}, 	\quad \rho\geq 1.
\end{equation}
Observe that $\Gamma(\rho)$ is non-empty since the original $\mypsi_k$ vectors satisfy the constraints.
The \emph{value} of an orthonormal representation is the quantity 
\begin{equation}
V(\{\tilde{\mypsi}_x\})=\min_{\myf}\max_x\log \frac{1}{|\braket{\tilde{\mypsi}_x}{\myf}|^2},
\end{equation}
where the minimum is over all unit norm vectors $\myf$. The optimal choice of the vector $\myf$ is called, with Lov\'asz, the handle of the representation. We call it $\myf$ to point out that this vector plays essentially the same role as the auxiliary output distribution $\mathbf{f}$ used in the sphere-packing bound of \cite{shannon-gallager-berlekamp-1967-1}. Due to space limitation, we cannot discuss this detail here; see the comment at the end of Section \ref{sec:connections}.

Call now $\vartheta(\rho)$ the minimum value over all representations of degree $\rho$, 
\begin{align}
\vartheta(\rho) & = \min_{\{\tilde{\mypsi}_x\} \in \Gamma(\rho)}V(\{\tilde{\mypsi}_x\}).
\end{align}

We have the following result.

\begin{theorem}
For any code of block-length $n$ with $M$ codewords and any $\rho\geq 1$ we have
\begin{equation}
\max_{m}\sum_{m'\neq m}(\braket{\myPsi_m}{\myPsi_{m'}})  \geq \frac{\left( M e^{-n\vartheta(\rho)} -1\right)^\rho}{(M-1)^{\rho-1}}
\end{equation}
\label{th}
\end{theorem}
\begin{corollary}
For the reliability function of a general DMC we have the bound
\begin{equation}
E(R)\leq 2\rho\,\vartheta(\rho), \qquad R>\vartheta(\rho),
\label{eq:bound_theta2_1}
\end{equation}
where the coefficient 2 can be removed if the channel is pairwise reversible.
\label{coroll}
\end{corollary}

\begin{proof}
For an input sequence $\bm{x}=(x_1,x_2,\ldots,x_n)$ call, in analogy with \eqref{eq:defPsi},   $\tilde{\myPsi}_{\bm{x}}=\tilde{\mypsi}_{x_1}\otimes\tilde{\mypsi}_{x_2}\otimes\cdots\tilde{\mypsi}_{x_n}$. Observe first that, for any two input sequences $\bm{x}$ and $\bm{x}'$, we have
\begin{align}
|\braket{\tilde{\myPsi}_{\bm{x}}}{\tilde{\myPsi}_{\bm{x}'}}| & = \prod_{i=1}^n|\braket{\tilde{\mypsi}_{x_i}}{\tilde{\mypsi}_{x_i'}}|\\
& \leq  \prod_{i=1}^n(\braket{{\mypsi}_{x_i}}{{\mypsi}_{x_i'}})^{1/\rho}\\
& = (\braket{\tilde{\myPsi}_{\bm{x}}}{\tilde{\myPsi}_{\bm{x}'}})^{1/\rho}
\end{align}
Furthermore, note that, for an optimal representation of degree $\rho$ with handle $\myf$, we have $|\braket{\tilde{\mypsi}_x}{\myf}|^2\geq e^{-\vartheta(\rho)}$, $\forall x$. Set now $\myF=\myf^{\otimes N}$. We then have
\begin{eqnarray}
|\braket{\tilde{\myPsi}_{\bm{x}}}{\myF}|^2 & = & \prod_{i=1}^n|\braket{\tilde{\mypsi}_{x_i}}{\myf}|^2\\
& \geq & e^{-n\vartheta(\rho)} \label{eq:tilted_vs_F}.
\end{eqnarray}
Let us first check how Lov\'asz's bound is obtained. Lov\'asz's approach is to bound the number $M$ of codewords with orthogonal state vectors, using the property that if $\tilde{\myPsi}_1,\tilde{\myPsi}_2,\ldots\tilde{\myPsi}_M$ form a set of orthonormal vectors, then
\begin{eqnarray}
1 & = & \|\myF\|_2^2 \\
& \geq & \sum_m |\braket{\tilde{\myPsi}_{m}}{\myF}|^2\\
& \geq & Me^{-n\vartheta(\rho)}.
\end{eqnarray}
Hence, if $M>e^{n\vartheta(\rho)}$, there are at least two non-orthogonal vectors in the set, say $|\braket{\tilde{\myPsi}_m}{\tilde{\myPsi}_{m'}}|^2>0$. But this implies that $(\braket{{\myPsi}_m}{{\myPsi}_{m'}})^2\geq|\braket{\tilde{\myPsi}_m}{\tilde{\myPsi}_{m'}}|^{2\rho}>0$. Hence, if $R>\vartheta(\rho)$, no zero-error code can exist. We still have the freedom in the choice of $\rho$ and it is obvious that larger values of $\rho$ can only give better results. Hence, it is preferable to simply work in the limit of $\rho\to \infty$ and thus build the representation $\tilde{\mypsi}_1, \tilde{\mypsi}_2,\ldots, \tilde{\mypsi}_K$ under the only constraint that $| \braket{\tilde{\mypsi}_x}{\tilde{\mypsi}_{x'}}| = 0$ whenever $| \braket{{\mypsi}_{x}}{{\mypsi}_{x'}}|=0$. This gives precisely Lov\'asz' result.

Now, instead of bounding $R$ under the hypothesis of zero-error communication, we want to bound the probability of error for a given $R>\vartheta(\rho)$.
Considering the tilted state vectors of the code, we can rewrite equation \eqref{eq:tilted_vs_F} as
\begin{eqnarray}
|\braket{\tilde{\myPsi}_{m}}{\myF}|^2  &  = & \bra{\myF}\left(\ket{\tilde{\myPsi}_m}\bra{\tilde{\myPsi}_m}\right)\ket{\myF}
\\ & \geq & e^{-n\vartheta(\rho)}.
\end{eqnarray}
The second expression above has the benefit of easily allowing averaging this expression over different codewords. So, we can average this expression over all $m$ and, defining the matrix $\myMat=\left(\tilde{\myPsi}_1,\ldots,\tilde{\myPsi}_M\right)/\sqrt{M}$, we get
\begin{equation}
\bra{\myF}\myMat \myMat^\dagger \ket{\myF}\geq e^{-n\vartheta(\rho)}.
\end{equation}
Since $\myF$ is a unit norm vector, this implies that the matrix $\myMat \myMat^\dagger$ has at least one eigenvalue larger than or equal to $e^{-n\vartheta(\rho)}$. This in turn implies that also the matrix  $\myMat^\dagger \myMat$ has itself an eigenvalue larger than or equal to $e^{-n\vartheta(\rho)}$, that is
\begin{equation}
\lambda_{\max}\left(\myMat^\dagger \myMat\right)\geq e^{-n\vartheta(\rho)}.
\end{equation}
It is known that for a given matrix $A$ with elements $\{A_{i,j}\}$, the following inequality holds
\begin{equation}
\lambda_{\max}(A)\leq \max_{i}\sum_j|A_{i,j}|.
\end{equation}
Using this inequality with $A=\myMat^\dagger \myMat$, since $A_{i,j}=\braket{\tilde{\myPsi}_i}{\tilde{\myPsi}_{j}}/M$, we get
\begin{eqnarray}
e^{-n\vartheta(\rho)} & \leq & 
\max_i\sum_j \frac{ |\braket{\tilde{\myPsi}_i}{\tilde{\myPsi}_{j}}| } {M}\\& = & \frac{1}{M}\left(1+\max_i\sum_{j\neq i}
|\braket{\tilde{\myPsi}_i}{\tilde{\myPsi}_{j}}|\right).
\end{eqnarray}
We then deduce
\begin{align}
\frac{Me^{-n\vartheta(\rho)}-1}{M-1} & \leq \max_i\frac{1}{M-1}\sum_{j\neq i}
|\braket{\tilde{\myPsi}_i}{\tilde{\myPsi}_{j}}|\\
&\leq \max_i\frac{1}{M-1}\sum_{j\neq i}
\left(\braket{\myPsi_i}{\myPsi_{j}}\right)^{1/\rho}\\
& \leq \max_i\left(\frac{1}{M-1}\sum_{j\neq i}
\braket{\myPsi_i}{\myPsi_{j}}\right)^{1/\rho},
\end{align}
where the last step is due to the Jensen inequality, since $\rho\geq 1$.
Extracting the sum from this inequality we obtain the inequality stated in the theorem.

To prove the corollary, simply note that 
\begin{align}
\max_{i\neq j}  \braket{\myPsi_i}{\myPsi_{j}} & \geq \max_i \frac{1}{M-1}\sum_{j\neq i}
|\braket{\myPsi_i}{\myPsi_{j}}|\\
& \geq \left(\frac{Me^{-n\vartheta(\rho)}-1}{M-1}\right)^{\rho}\label{eq:maxprodscalavrg}\\
& \geq \left( e^{-n\vartheta(\rho)}-e^{-nR}\right)^{\rho}.
\end{align}
The bound is trivial if $R\leq \vartheta(\rho)$. For $R>\vartheta(\rho)$, instead, the second term in the parenthesis decreases exponentially faster than the first, which leads us to the conclusion that 
\begin{equation}
-\frac{1}{n}\log \max_{m\neq m'} \braket{{\myPsi}_m}{{\myPsi}_{m'}}\leq \rho \vartheta(\rho) + o(1).
\end{equation}
The bounds in terms of $E(R)$ are then obtained by simply taking the limit $n\to\infty$ and using the bound \eqref{eq"pemaxprodscal1}.
\end{proof}
%
%

We close this section with a comment on the computation of the function $\vartheta(\rho)$. There is no essential difference with respect to the evaluation of the Lov\'asz theta function. The optimal representation $\{\tilde{\mypsi}_x\}$ for any fixed $\rho$, can be obtained by solving a semidefinite optimization problem. If we consider the $(K+1)\times (K+1)$ Gram matrix
\begin{equation}
G=[\tilde{\mypsi}_1,\ldots,\tilde{\mypsi}_K, \myf]^\dagger[\tilde{\mypsi}_1,\ldots,\tilde{\mypsi}_K, \myf]
\end{equation}
we note that finding the optimal representation amounts to solving the problem
\begin{equation}
\begin{array}{lrcl}
& \max  & V & \\
\mbox{s.t.} & G(k,K+1) & \geq  &V, \quad \forall k\leq K\\
& G(k,k) & =  & 1, \quad \forall k\\
& G(k,i) & \leq & \braket{\mypsi_k}{\mypsi_i}^{1/\rho}\\
& & & 1\leq k <K,\, k<i\leq K\\
 &  G & \mbox{is} & \mbox{positive semidefinite}
 \end{array}
\end{equation}
The solution to this problem is $V=\vartheta(\rho)$ and both the representation vectors $\{\tilde{\mypsi}_x\}$ and the handle $\myf$ can be obtained by means of the spectral decomposition of the optimal $G$ found.

\section{Connections with other results in channel theory}
\label{sec:connections}
A first important comment abount $\vartheta(\rho)$ concerns the result obtained for $\rho=1$; the value $\vartheta(1)$ is in fact simply the cut-off rate of the channel. Indeed, for $\rho=1$, we can without loss of generality use the obvious representation $\tilde{\mypsi}_x=\mypsi_x, \forall x$, since any different optimal representation will simply be a rotation of this (or an equivalent description in a space with a different dimension). In this case, all the components of all the vectors $\{\tilde{\mypsi}_x\}$ are non-negative and this easily implies that the optimal $f$ can as well be chosen with non-negative components, since changing a supposedly negative component of $\myf$ to its absolute value can only improve the result.
Thus, $\myf$ can be written as the square root of a probability distribution $Q$ on $\mathcal{Y}$, and we have 
\begin{eqnarray}
\vartheta(1) & = & \min_{\myf}\max_x\log \frac{1}{|\braket{{\mypsi}_x}{\myf}|^2}\\
& = & \min_{Q} \max_x \left(-2\log \sum_{y}\sqrt{W(y|x)Q(y)}\right)
\end{eqnarray}
where the minimum is now over all probability distributions $Q$.
As observed by Csisz\'ar \cite[Proposition 1, with $\alpha=1/2$]{csiszar-1995}, this expression equals the cut-off rate
 $R_1$ of the channel defined as
\begin{eqnarray*}
R_1 & = & \max_{P}- \log \sum_{x,x'}P(x)P(x')\left(\sum_y \sqrt{W(y|x)W(y|x')}\right)\\
& = & \max_{P}- \log \sum_{x,x'}P(x)P(x') \braket{\mypsi_x}{\mypsi_{x'}}
\end{eqnarray*}

Another important characteristic of the function $\vartheta(\rho)$ is observed in the limit $\rho\to\infty$. In the limit, the only constraint on the representations is that $| \braket{\tilde{\mypsi}_x}{\tilde{\mypsi}_{x'}}| = 0$ whenever $| \braket{{\mypsi}_{x}}{{\mypsi}_{x'}}|=0$. Hence, when $\rho\to\infty$, the set of possible representations is precisely the same considered by Lov\'asz \cite{lovasz-1979}, and we thus have $\vartheta(\rho)\to\vartheta$ as $\rho\to \infty$. So, the value of $\vartheta(\rho)$ moves from the cut-off rate $R_1$ to the Lov\'asz bound $\vartheta$ when $\rho$ varies from $1$ to $\infty$.
This clearly implies that the bound of Corollary \ref{coroll} is finite for all $R>\vartheta$ and thus it allows to bound the zero-error capacity of the channel as
\begin{eqnarray}
C_0 & \leq & \lim_{\rho\to\infty}\vartheta(\rho)\\
& = & \vartheta.
\end{eqnarray}

In general, the function $\vartheta(\rho)$ turns out to be strongly related to the coefficient $E_x(\rho)$  used in the expurgated bound of Gallager \cite{gallager-1965} and defined, using our definition of $\ket{\mypsi_x}$,  as
\begin{multline}
E_x(\rho)=
\max_{P} \left[-\rho \log \sum_{x,x'}P(x)P(x') (\braket{\mypsi_x}{\mypsi_{x'}})^{1/\rho}\right].
\end{multline}
In order to present this relation, it is instructive to consider first the so called \emph{non-negative definite} channels as defined by Jelinek \cite{jelinek-1968}. These are channels for which the matrix $\tilde{C}$ with $(i,j)$ element $\tilde{c}_{i,j}=(\braket{\mypsi_i}{\mypsi_j})^{1/\rho}$ is positive semidefinite for all $\rho\geq 1$. For example, the binary symmetric channel (BSC) is non-negative definite. It was proved by Jelinek that, for these channels, the expurgated coefficient $E_x^{(n)}(\rho)$ computed over the $n$-fold extensions of the channel (and normalized to $n$) has the same value as $E_x(\rho)$.
It is also known that for these channels, the inputs can be partitioned in subsets such that all pairs of symbols from the same subset are confusable and no pair of symbols from different subsets are confusable. The zero error capacity in this case is simply the logarithm of the number of such subsets.
For these channels, since the matrix $\tilde{C}$ is positive semidefinite, there exists a set of vectors $\tilde{\mypsi}_1,\tilde{\mypsi}_2,\ldots,\tilde{\mypsi}_K$ such that $\braket{\tilde{\mypsi}_i}{\tilde{\mypsi}_j}=\tilde{c}_{i,j}$, that is, for all $\rho\geq 1 $, representations of degree $\rho$ exist that satisfy all the constraints with equality. 
In this case, the equivalence with the cut-off rate that we have seen for $\rho=1$ can be in a sense extended to other $\rho$ values. 
It can be proved \cite[Th. 9]{dalai-QSP-2012} that we can write 
\begin{align}
\vartheta(\rho) = &\min_{\myf}\max_x\log \frac{1}{|\braket{\myf}{\tilde{\mypsi}_x}|^2}\\
 & =  \max_{P}\left[- \log \sum_{x,x'}P(x)P(x') \braket{\tilde{\mypsi}_x}{\tilde{\mypsi}_{x'}}\label{eq:alt_def_theta}\right]\\
  & =  \max_{P} \left[- \log \sum_{x,x'}P(x)P(x') (\braket{{\mypsi}_x}{{\mypsi}_{x'}})^{1/\rho}\label{eq:alt_def_theta2}\right]
\end{align}
Hence, under such circumstances, we find that $\vartheta(\rho)=E_x(\rho)/\rho $. For example, for the BSC with transition probability $\varepsilon$, we have
\begin{equation}
\vartheta(\rho)=-\log \left( \frac{1}{2}+\frac{1}{2}[4\varepsilon(1-\varepsilon)]^{1/2\rho}\right).
\end{equation}
In general, for non-negative definite channels, the bound of Corollary \ref{coroll} is obtained by drawing the curve parameterized as $(E_x(\rho)/\rho, 2 E_x(\rho))$ in the $(R,E)$ plane. Thus, it is seen that that bound is loose in general. It is however somehow tight in the sense that, in this particular case,
\begin{align}
C_0 = \lim_{\rho\to\infty} \vartheta(\rho),
\end{align}
(which is however trivial) and, if $C_0=0$, the bound gives
\begin{eqnarray}
E(0) & \leq & \lim_{\rho\to\infty} 2\rho\,\vartheta(\rho)\\
& = & \lim_{\rho\to\infty} 2E_x(\rho)\\
& = & 2 E_{ex}(0),
\end{eqnarray}
where $E_{er}(R)$ is Gallager's expurgated lower bound to $E(R)$.
If the channel is pairwise reversible, this can then be improved to $E(R)\leq E_{ex}(0)$, which is obviously tight.

For general channels with a non-trivial zero-error capacity, like for example any channel whose confusability graph is a pentagon, what happens is that the matrix $\tilde{C}$ is in general positive semidefinite only for values of $\rho$ in a range $[1,\bar{\rho}]$ and then it becomes not positive semidefinite for some  $\rho>\bar{\rho}$. This implies that for $\rho>\bar{\rho}$, representations that satisfy all the constraints with equality do not exist in general. In this case, the two expressions in equations \eqref{eq:alt_def_theta} and \eqref{eq:alt_def_theta2} are no more equal and in general they could both differ from $\vartheta(\rho)$. If all the values $\braket{\tilde{\mypsi}_x}{\tilde{\mypsi}_{x'}}$ are nonnegative\footnote{We conjecture that the optimal representation, in terms of Lov\'asz's definition of \emph{value}, always satisfies this condition. We have not yet investigated this aspect, but have never found a counterexample.}, however, then it can be proved that the expression in \eqref{eq:alt_def_theta} equals $\vartheta(\rho)$ \cite[Th. 9]{dalai-QSP-2012}. In this case, we see the interesting difference between $\vartheta(\rho)$ and $E_x(\rho)/\rho$. The two 	quantities follow respectively \eqref{eq:alt_def_theta} and \eqref{eq:alt_def_theta2}. When $\rho\to\infty$, $\vartheta(\rho)$ tends to $\vartheta$, an upper bound to $C_0$. The value $E_x(\rho)/\rho$ instead is known to converge to the independence number of the confusability graph of the channel \cite{korn-1968}, a lower bound to  $C_0$.

More generally, if $\braket{\tilde{\mypsi}_x}{\tilde{\mypsi}_{x'}}\geq 0$,  $\forall x,x'$ , since $\vartheta(\rho)$ is given by equation \eqref{eq:alt_def_theta}, it is an upper bound to \eqref{eq:alt_def_theta2} and thus to $E_x(\rho)/\rho$. It can then also be proved that \eqref{eq:alt_def_theta} is multiplicative, in this case, over the $n$-fold tensor power of the representation $\{\tilde{\mypsi}_k\}$. This implies that, for all $n$, $\vartheta(\rho)$ is an upper bound to the (normalized) expurgated bound $E_x^{(n)}(\rho)/\rho$ computed for the $n$-fold memoryless extension of the channel. That is, $\vartheta(\rho)$ generalizes $\vartheta$ in the sense that, in the same way as
\begin{equation}
\vartheta \geq \sup_{n}\lim_{\rho\to \infty} \frac{E_x^{(n)}(\rho)}{\rho}=C_0,
\end{equation} 
also
\begin{equation}
\vartheta(\rho) \geq \sup_{n}\frac{E_x^{(n)}(\rho)}{\rho}.
\end{equation} 
The discussion of this point with generality requires some technicalities and will hopefully be given in a future work (see footnote 2).

Is is worth pointing out that, for some channels, the optimal representation may even stay fixed for $\rho$ larger than some given finite value $\rho_{\max}$ and $\vartheta(\rho)$ is thus constant for $\rho\geq \rho_{\max}$ (in this case, the bounds are useless for $\rho>\rho_{\max}$). This happens for the famous example for the noisy typewriter channel with five inputs and crossover probability $1/2$. In this case $\bar{\rho}=\rho_{\max}\approx 2.88$; as shown in Fig. \ref{fig}, for $\rho<\rho_{\max}$ we have $\vartheta(\rho)=E_x(\rho)/\rho$ while, for $\rho\geq \rho_{\max}$, $\vartheta(\rho)=C_0=\log\sqrt{5}$.

\begin{figure}
\includegraphics[width=\linewidth]{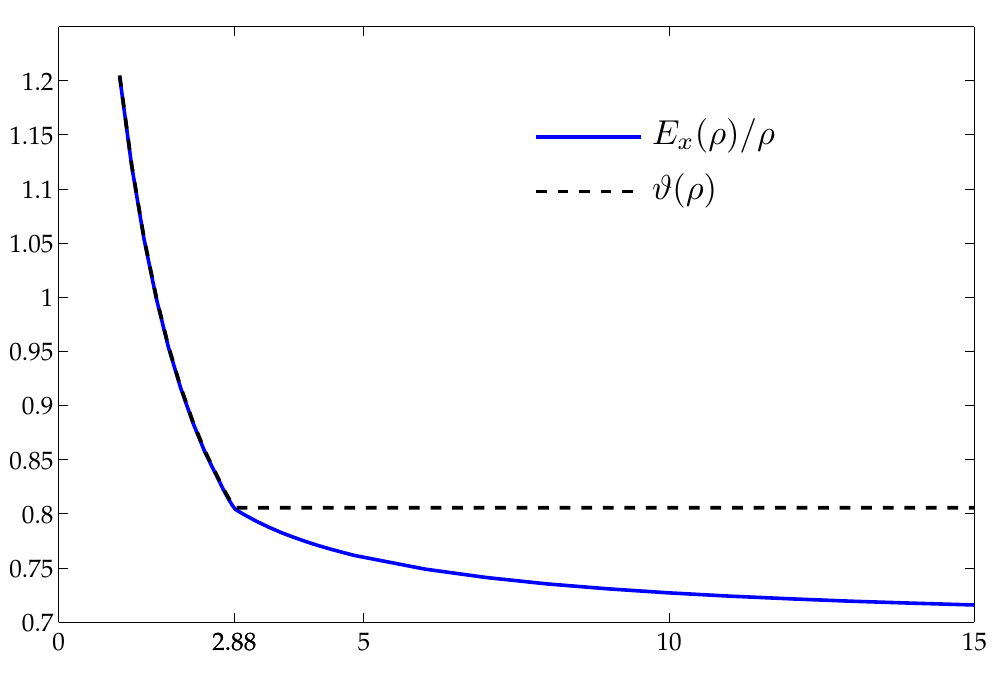}
	\caption{Plot of $\vartheta(\rho)$ and $E_x(\rho)/\rho$ for the noisy typewriter channel with five inputs and crossover probability $1/2$.}
\label{fig}
\end{figure}

%
We close the paper with a comment on the relation between the results presented here and some recent results in quantum information theory. In that context, it is revealed \cite{dalai-ISIT-2013a} that Lov\'asz's idea is intimately related to the sphere-packing bound of \cite{shannon-gallager-berlekamp-1967-1}.
The bound to $E(R)$ derived here is a special case of a more general bound that can be derived by properly applying the sphere-packing bound for classical-quantum channels \cite{dalai-ISIT-2012}, \cite{dalai-QSP-2012}.
In particular, while the construction of the representation $\{\tilde{\mypsi}_x\}$ was introduced here as a purely mathematical trick to bound $E(R)$, this procedure can be interpreted in the context of classical-quantum channels as a natural way to bound $E(R)$ by comparing the original channel with an auxiliary one. 


\bibliographystyle{IEEEtran}


\end{document}